%% file: arXiv.tex
\documentclass[11pt]{article}
\usepackage{url}
\usepackage{cite}
\usepackage{times}
\usepackage{amsmath,xspace,amssymb,epsfig,amsthm}
\usepackage{fullpage}
\usepackage{algorithm}
\usepackage{algorithmic}
\usepackage{color}
\usepackage{subfigure}
\usepackage[small]{caption}
\usepackage{enumitem}

\newtheorem{Thm}{Theorem}

\newtheorem{Cor}[Thm]{Corollary}
\newtheorem{Prop}[Thm]{Proposition}

\newtheorem{Def}{Definition}
\newtheorem{Fact}{Fact}

\newcommand\R{\mathsf{R}}

\newcommand\N{\mathbb{N}}

\newcommand{\Cut}{{\sc AustinCut}\xspace}
\newcommand{\AlgTree}{{\sc AlgTree}\xspace}
\newcommand{\AllocationTree}{{\sc AllocationTree}\xspace}

\newcommand{\AlgTreeLike}{{\sc AlgDescendantGraph}\xspace}

\begin{document}

\title{Networked Fairness in Cake Cutting}
% \thanks{This work was supported by
% Australian Research Council DECRA DE150100720 and Research
% Grants
% Council of the Hong Kong S.A.R. (Project no. CUHK14239416)}}

\author{Xiaohui Bei
\footnote{Nanyang Technological University, Singapore.}
\and Youming Qiao
\footnote{University of Technology Sydney, Australia.}
\and Shengyu Zhang
\footnote{The Chinese University of Hong Kong, Hong Kong.}
% $^1$School of Physical and Mathematical Sciences, Nanyang Technological University\\
% $^2$Centre for Quantum Software and Information, \\
% $^3$Department of Computer Science and Engineering, \\
% $^1$xhbei@ntu.edu.sg, $^2$Youming.Qiao@uts.edu.au, $^3$syzhang@cse.cuhk.edu.hk
}

% \author{
% Xiaohui Bei \\ Nanyang Technological University \\ xhbei@ntu.edu.sg \And
% Youming Qiao \\ University of Technology Sydney \\ Youming.Qiao@uts.edu.au \AND
% Shengyu Zhang \\ The Chinese University of Hong Kong \\ syzhang@cse.cuhk.edu.hk}
\date{}
\maketitle

\begin{abstract}
We introduce a graphical framework for fair division in cake cutting, where
comparisons between agents are limited by an underlying network structure. We
generalize the classical fairness notions of envy-freeness and proportionality to
this graphical setting. Given a simple undirected graph $G$, an allocation is
\textit{envy-free on $G$} if no agent envies any of her neighbor's share, and is
\textit{proportional on $G$} if every agent values her own share no less than the
average among her neighbors, with respect to her own measure. These
generalizations open new research directions in developing simple and efficient
algorithms that can produce fair allocations under specific graph structures.

On the algorithmic frontier, we first propose a moving-knife algorithm that
outputs an envy-free allocation on trees. The algorithm is significantly simpler
than the discrete and bounded envy-free algorithm recently designed
in~\cite{aziz2016discrete} for complete graphs.
Next, we give a discrete and bounded algorithm for computing a proportional
allocation on descendant graphs, a class of graphs by taking a rooted tree and
connecting all its ancestor-descendant pairs.
\end{abstract}

\input{intro}

\section{Our Framework}
% For $n\in\N$, $[n]$ denotes the set $\{1,2,\ldots, n\}$.

In a cake cutting instance, there are $n$ agents to share a cake, which is
represented by the interval $[0,1]$.
Each agent $i$ has an integrable, non-negative density function $v_i: [0, 1] \mapsto \R$.
A piece $S$ of the cake is the union of finitely many disjoint intervals of $[0,
1]$. The valuation of agent $i$ for a piece $S$ is
$V_i(S) = \int_Sv_i(x) dx$.

An allocation of the cake is a partition of $[0, 1]$ into $n$ disjoint pieces, denoted by $A = (A_1, \ldots, A_n)$, such that agent $i$ receives piece $A_i$, where all pieces are disjoint and $\bigcup_i A_i = [0,1]$.

\subsection{Fairness Notions on Graphs}

First we review two classic fairness notions in resource allocation.

\vspace{0.5em}
\noindent{\bf Envy-Freeness.} An allocation $A$ is called {\it envy-free} if for all $i \neq j$, $v_i(A_i) \geq v_i(A_j)$.

\noindent{\bf Proportionality.} An allocation $A$ is called {\it proportional} if for all $i$, $v_i(A_i) \geq \frac1n v_i([0,1])$.
\vspace{0.5em}

To account for a graphical topology, we generalize the above definitions by assuming a network structure over all agents. That is, we assume an undirected simple graph $G = (V,E)$ in which each vertex $i$ represents an agent. Let %$\deg(i)$ denote the degree of agent/vertex $i$, and
$N(i)$ be the set of neighbors of agent $i$
in $G$.

\begin{Def}[Envy-free on networks]
  An allocation $A$ is called {\it envy-free on a network $G = (V,E)$} if for all $i$ and
  all $j \in N(i)$,
  $v_i(A_i) \geq v_i(A_j)$.
\end{Def}

\begin{Def}[Proportional on networks]
  An allocation $A$ is called {\it proportional on a network $G$} if for all $i$,
  $v_i(A_i) \geq \frac1{|N(i)|}\sum_{j \in N(i)}v_i(A_j).$
\end{Def}

Both definitions can be viewed as generalizations of the original fairness concepts, which correspond to the case of $G$ being the complete graph. %Indeed, when $G$ is the complete graph, both definitions on networks are equivalent to their original definitions.

For envy-freeness, we have the following properties.
\begin{Fact}\label{prop:subgraph}
  Any allocation that is envy-free on a graph $G$ is also envy-free on any
  subgraph $G' \subseteq G$.
\end{Fact}
\begin{Cor}
  An envy-free allocation is also envy-free on any graph $G$.
\end{Cor}
The corollary implies that the envy-free protocol recently proposed by Aziz and
Mackenzie~\cite{aziz2016discrete} would also give an envy-free allocation on any graph.
However, the protocol is highly involved and requires a large number of queries and cuts, which naturally raises the question of designing simpler protocols for
graphs with special properties.

\medskip
\noindent\fbox{
\parbox{\textwidth}{
  {\bf Goal 1:} Design simple protocols that could produce envy-free allocations on certain types of graphs.
}
}
\medskip

While envy-freeness is closed under the operation of edge removal, proportionality is not, as removing edges changes the set of neighbors and thus also the average of neighbors' values. %In particular, a proportional allocation may lose its proportionality when restricted on a particular graph.
Therefore, although a number of proportional protocols are known, they do not readily translate to proportional allocations on subgraphs in any straightforward way.

At a first glance, even the existence of a proportional allocation on all graphs is not obvious.
Interestingly, the existence can be guaranteed from that of an envy-free
allocation, since any envy-free allocation is also proportional on the same
network.
Therefore an envy-free protocol for complete graphs would also produce a proportional
allocation on any graph. In light of this, we turn to
the quests for simpler
proportional protocols on interesting graph families as in the envy-free case.

\medskip
\noindent\fbox{
\parbox{\textwidth}{
{\bf Goal 2:} Design simple protocols that could produce proportional allocations
on certain types of graphs.
}
}
\medskip

\subsection{Motivations for the Two Graph Families}\label{subsec:motivation}
In this paper we focus on two graph families, trees and descendant graphs. %We explain why these two graph families are interesting in this context.
%Let us motivate the consideration of trees first. It is clear that,
Note that for both envy-freeness and proportionality, we only need to consider connected graphs (since otherwise we can focus on the easiest connected component). %By Fact~\ref{prop:subgraph}, an envy-free allocation protocol on a connected graph $G$ induces an envy-free protocol on any of its spanning trees. Therefore, it seems fair to say that
Without an envy-free algorithm for trees, there would be no chance to design envy-free protocols for more complicated graph families (see Fact~\ref{prop:subgraph}). %Therefore we address the case of trees in the envy-free setting in Section~\ref{sec:envy-freeness-trees}.

%We then turn to
For descendant graphs, recall that a descendant graph is obtained by
connecting all ancestor-descendant pairs of a rooted tree (cf.
Section~\ref{sec:prop-tree} for a formal definition). %The interests in such graphs come from both practical and theoretical viewpoints. From a practical viewpoint,
Descendant graphs can be used to model the relations among members in
an extended family, where each edge connects a pair of ancestor and descendant. It
can also model the management hierarchy in a company, where each edge connects a
superior and a subordinate. From a theoretical viewpoint, we note that
for any undirected graph $G$, if we run a depth-first search, then all edges in $G$ are either edges on the DFS tree $T$ or a back edge of $T$. Thus $G$ is a subgraph of the upward closure $T^{uc}$ of $T$, where $T^{uc}$ is obtained from $T$ by connecting all (ancestor, descendant) pairs. Note that $T^{uc}$ is a descendant graph. Therefore, if there is an envy-free allocation protocol for the family of descendant graphs, then there is an envy-free allocation protocol for any graph.
%DFS tree can be embedded as a subgraph of a descendant graph by the descendant graph with respect to a depth-first search tree of $G$. This yields the following corollary to Fact~\ref{prop:subgraph}.
%\begin{Cor}\label{cor:descendant}
%An envy-free allocation protocol for the family of descendant graphs yields an envy-free allocation protocol for any graph.
%\end{Cor}
%Corollary~\ref{cor:descendant}
This indicates that envy-freeness for descendant graphs
may be hard. Interestingly, if we relax to proportionality, we do get a protocol
on this family, as shown in Section~\ref{sec:prop-tree}.

\section{Envy-Freeness on Trees}
\label{sec:envy-freeness-trees}
We first present an algorithm that produces an envy-free allocation on trees. The
algorithm makes use of the {\em Austin moving-knife procedure}~\cite{austin1982sharing}
as a subroutine.
Given a rooted tree $T$, let $|T|$ denote the number of vertices. For any vertex
$i$ in $T$, denote by $T(i)$ the subtree rooted at vertex $i$.

\begin{Def}[\cite{austin1982sharing}]% [Austin moving-knife procedure.]
An {\em Austin moving-knife procedure \Cut$(i, j, n, S)$} takes two agents $i, j$,
an integer $n>0$ and a subset $S$ of cake as input parameters, and outputs a
partition of $S$ into $n$ parts, such that \emph{both} agents value these $n$
pieces as all equal, each of value exactly $1/n$ fraction of the value of $S$. An Austin procedure needs $2n$ cuts.
\end{Def}

It should be noted that this is a continuous procedure, and it is not known that
whether it can be implemented by a discrete algorithm. With that being said, to the best of
our knowledge, it is not known whether Austin's moving knife procedure can help to
obtain an algorithm achieving envy-freeness (on complete graphs) that is
simpler than the algorithm in \cite{aziz2016discrete}.

Now we are ready to present our allocation algorithm in Algorithm \AllocationTree, which calls a sub-procedure given in Algorithm \AlgTree.
\begin{algorithm}
  \caption{\AlgTree$(T, r, (A_1, \ldots, A_n))$}
  \label{alg:tree-sub}
  \begin{algorithmic}[1]
    \REQUIRE  Tree $T$ with root vertex $r$, size of tree $n$, an allocation $A = (A_1, \ldots, A_n)$.

    \FOR {each immediate child $i$ of root $r$}
        \STATE Among all remaining pieces of $A_1, \ldots, A_n$, pick $|T(i)|$ pieces that agent $i$ values the highest.
        \STATE Let $S_i$ denote the union of these $|T(i)|$ pieces.
    \ENDFOR
    \STATE \label{step: to root} Allocate the remaining one piece to $r$.
    \FOR {each child $i$ of root $r$}
        \STATE \label{step:Austin-cut} Apply \Cut$(i, r, |T(i)|, S_i)$ to divide $S_i$ into $|T(i)|$ equal parts $X_1, \ldots, X_{|T(i)|}$ for $i$ and $r$.
        \STATE Run \AlgTree$(T(i), i, (X_1, \ldots, X_{|T(i)|}))$.
    \ENDFOR
  \end{algorithmic}
\end{algorithm}

\begin{algorithm}
	\caption{\AllocationTree$(T, r)$}
	\label{alg:tree-main}
	\begin{algorithmic}[1]
		\REQUIRE  Tree $T$ with root vertex $r$.

		\STATE Let $r$ cut the cake into $n=|T|$ equal parts $(A_1, \ldots, A_n)$ with respect to her valuation.
		\STATE Run \AlgTree$(T, r, (A_1, \ldots, A_n))$.
	\end{algorithmic}
\end{algorithm}

\begin{Thm}
  For any $n$-node tree $T$ with root $r$, Algorithm \AllocationTree$(T, r)$ outputs an allocation that is envy-free on $T$ with $O(n^2)$ cuts.
  % \shengyu{Should we mention running time here? How to measure running time of the moving-knife procedure?}
  %Given a tree $T$ rooted at vertex $r$, first let $r$ cut the cake into $n=|T|$ equal pieces $(A_1, \ldots, A_n)$, then \AlgTree$(T, r, (A_1, \ldots, A_n))$ returns an envy-free allocation on tree $T$.
\end{Thm}
\begin{proof}
	We will prove this by induction on the height of the tree. The base case is when the tree has only one vertex, the root $r$. In this case $n=1$ and the algorithm simply gives the whole set $A_1$ to $r$, and the envy-free property holds trivially. Now we assume that the theorem holds for trees of height $d$ and will prove it for any tree $T$ of height $d+1$. First, the root $r$ receives exactly $1/n$ as she cuts the cake into $n$ equal pieces (in Algorithm \AllocationTree) and gets one of them (in line \ref{step: to root} of Algorithm \AlgTree). In addition, in $r$'s valuation, each child $i$ (of $r$) gets $|T(i)|/n$ fraction of the total utility (the first \textbf{for} loop in Algorithm \AlgTree),
  and cuts it into $|T(i)|$ equal pieces (with respect to $r$'s valuation, in line \ref{step:Austin-cut} in Algorithm \AlgTree) and finally gets one of them. So $r$ thinks that child $i$'s share is also worth $1/n$ utility, same as hers. Thus $r$ does not envy any of her children.

	Also note that all children of $r$ pick pieces before $r$ does, so each child
	$i$ values her part $S_i$ at least $|T(i)|$ times of what $r$ gets. Then the
	second \textbf{for} loop in Algorithm \AlgTree cuts this part into $|T(i)|$
	equal pieces (with respect to child $i$'s valuation, in line
	\ref{step:Austin-cut} in Algorithm \AlgTree).
  Since child $i$ gets one of these pieces (in line \ref{step: to root} of the recursive call of Algorithm \AlgTree) she views this piece at least $\frac{1}{|T(i)|}\cdot |T(i)| = 1$ times what $r$ gets. Namely, child $i$ does not envy root $r$.

	Among vertices inside each subtree rooted at child $i$ of $r$, no
	envy occurs by the inductive hypothesis. Putting everything together, we can
	see that there is no envy between any pair of connected vertices.

	Finally we analyze the number of cuts of the algorithm. Recall that \Cut$(i,j,n,S)$ makes at most $2n$ cuts. When each vertex $u$ is being processed in \AlgTree, it makes at most $\sum_{i: \text{child of }u} 2|T_i| \leq 2n$ cuts. Thus in total the algorithm requires $O(n^2)$ cuts.
  % the recursion $p(|T|) = 1 + \sum_{i: \text{child of root}} |T_i| p(|T_i|)$. It is then easy to verify that $p(n) \le \sum_{i=0}^{n-1} (n-1)(n-2)\ldots (n-i) \le n!.$
\end{proof}

\section{Proportionality on Descendant Graphs}
\label{sec:prop-tree}

We first define descendant graphs formally.

\begin{Def}
  An undirected graph $G = (V,E)$ is called a {\em descendant graph on a rooted
  tree
  $T$},
  if $T$ is a spanning tree of $G$, and there exists a vertex $r$ such that for
  any two vertices $i$ and $j$, $(i,j)\in E$ if and only if $(i,j)$ is an
  ancestor-descendant pair on tree $T$ rooted at $r$. %When viewing $T$ as a
  %tree rooted at $r$, it satisfies that for every pair of vertices $i$ and $j$,
  %$(i,j) \in G$ if and only if $i$ and $j$ are ancestor-descendant pairs in $T$.
\end{Def}
In other words, a descendant graph $G$ of a rooted tree $T$ can be obtained from $T$
by connecting all of ancestor-descendant pairs.

\begin{figure}[!ht]
	\centering
    \subfigure[a tree $T$ with the topmost node being its root ]{
        % \label{fig:fvalues}
        \includegraphics[height=1.8in]{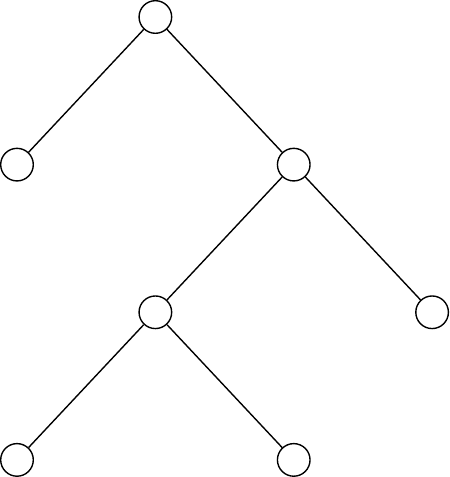}
    }
    \qquad\qquad\qquad\qquad
    \subfigure[the descendant graph on $T$]{
        % \label{fig:exec}
        \includegraphics[height=1.8in]{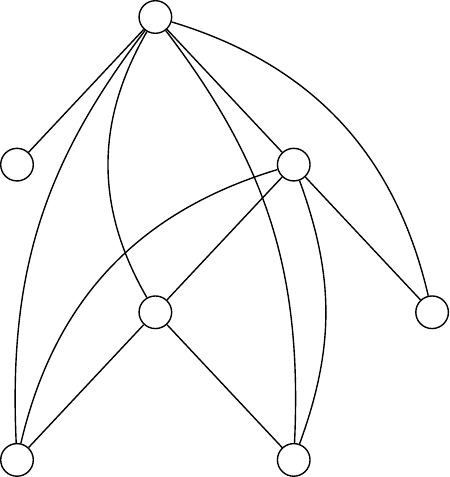}
    }
	\caption{A 7-node tree and its descendant graph}
\end{figure}

% Let $T$ be a rooted tree with the root $r$. We define the tree-like graphs based
% on $T$, denoted as $\treelike{T}$, by adding more edges to $T$ as follows: for
%two
% vertice $u\neq v$, if $u$ is on the path from $v$ to $r$, then add an edge
%between
% $u$ and $v$. Here $u$ can be $r$ as well.

We now present an algorithm that produces a proportional allocation on descendant
graphs.
The idea of the algorithm can be described as a process of collecting and distributing: We start with the root vertex holding the whole cake, and process the tree in top-down fashion.
Each vertex $v$, when being processed, applies a three-step procedure. (1) {\em Collect phase:} agent $v$ collects all cake pieces that she has received (from her ancestors); (2) {\em Cut phase:} agent $v$ cuts these cakes into $f(v)$ equal pieces according to her own evaluation, for some function $f(v)$ to be defined later; (3) {\em Distribute phase:} let each descendant of $v$ pick certain number of these pieces that they value the highest.

To formally describe the algorithm, we shall need the following notation. For a
vertex $v \in T$, let $d(v)$ denote the depth of vertex $v$ in tree $T$ (i.e. the number of edges on the unique path from $v$ to the root), $d = \max_v d(v)$ denote the depth of $T$.
% Let $d$ be the depth of $T$, defined
% as the maximum of $d(v, T)$ over all $v \in T$. Let $T(v)$ denote the subtree rooted at $v$.

% Fix a vertex $v\in T$ with $k=d(v, T)$ ($0\leq k\leq d$), and let
% $n_i=n_i(v):=|\{u\in T(v) : d(u, T(v))=i\}|$.
We then define a function $f:T\to \N$ as
\begin{equation}\label{eq:def_fv}
f(v)=\left(\frac{d(v)+|T(v)|}{d(v)+1}\right) \cdot d!.
%
%\frac{1}{k+1}\cdot n_1\cdot d!+ \frac{1}{k+1}\cdot
%n_2\cdot d! + \dots +
%\frac{1}{k+1}\cdot n_{d-k}\cdot
%d!.
\end{equation}
We first make some easy observations from this definition.
\begin{itemize}
    \item When $d(v)\geq 1$, $d(v)$ divides $f(v)$. This guarantees that in the algorithm
    described below, each distribute step reallocates an integral number of slices.
    \item $f(r) = |T| \cdot d!$ for root $r$ of the tree.
    \item $f(v) = d!$ for every leaf $v$ of the tree.
\end{itemize}

\begin{Prop}\label{prop:induct_fv}
    Function $f(v)$ satisfies
% Suppose $v$ is of depth $k$ in tree $T$. For $k'=k+1, k+2, \dots, d$, let $N_{k'}$ be $v$'s
% descendants at depth $k'$. Then
% \begin{multline}\label{eq:induct_fv}
% f(v)=d!+\frac{1}{k+1}\sum_{u_1\in
% N_1}f(u_1)+\frac{1}{k+2}\sum_{u_2\in N_2}f(u_2) +\\\dots
% +\frac{1}{d}\sum_{u_{d-k}\in
% N_{d-k}}f(u_{d-k}).
% \end{multline}
\begin{equation}\label{eq:induct_fv}
f(v)=d!+\sum_{\substack{u \in T(v) \\ u \neq v}}\frac{f(u)}{d(u)}
\end{equation}
\end{Prop}

\begin{proof}
    Define $g(v) = \frac{f(v)}{d!}$. Thus $g(v) = \frac{d(v)+|T(v)|}{d(v)+1}$.

    For each vertex $v$, we have

    \begin{align*}
        &  \sum_{\substack{u \in T(v)\\ u \neq v}}{\frac{g(u)}{d(u)}} = \sum_{\substack{u \in T(v)\\ u \neq v}}\frac{1}{d(u)}\left(\frac{d(u)+|T(u)|}{d(u)+1}\right) \\
        = & \sum_{\substack{u \in T(v)\\ u \neq v}}\left(\frac{1}{d(u)+1} + \frac{|T(u)|}{d(u)(d(u)+1)}\right) \\
        = & \sum_{\substack{u \in T(v)\\ u \neq v}}\left(\frac{1}{d(u)+1} + \frac{\sum_{u' \in T(u)}1}{d(u)(d(u)+1)}\right) \\
        = & \sum_{\substack{u \in T(v)\\ u \neq v}}\frac{1}{d(u)+1}+\sum_{\substack{u \in T(v)\\ u \neq v}}\sum_{u'\in T(u)}\frac{1}{d(u)(d(u)+1)}\\
        = & \sum_{\substack{u \in T(v)\\ u \neq v}}\frac{1}{d(u)+1}+\sum_{\substack{u' \in T(v)\\ u' \neq v}}\underbrace{\sum_{\substack{u\in T(v)\\u \neq v\\u'\in T(u)}}\frac{1}{d(u)(d(u)+1)}}_{X}.
    \end{align*}

    Now let us take a closer look at term $X$. Note that it is summed over all vertices $u$ at the path between $v$(exclusive) and $u'$(inclusive). These vertices have depth $d(v)+1, d(v)+2, \ldots, d(u')$ in tree $T$, respectively. Thus we have
    \begin{align*}
        X & = \sum_{k=d(v)+1}^{d(u')}\frac{1}{k(k+1)} = \sum_{k=d(v)+1}^{d(u')}\left(\frac1k - \frac1{k+1}\right) \\
        & = \frac1{d(v)+1} - \frac1{d(u')+1}.
    \end{align*}
    Plugging this back to the previous formula gives
    \begin{align*}
        & \sum_{\substack{u \in T(v)\\ u \neq v}}\frac{1}{d(u)+1}+\sum_{\substack{u' \in T(v)\\ u' \neq v}}
        \left(\frac1{d(v)+1} - \frac1{d(u')+1}\right) \\
        = & \sum_{\substack{u \in T(v)\\ u \neq v}}\frac{1}{d(v)+1} = \frac{|T(v)|-1}{d(v)+1}.
    \end{align*}
    Therefore, $$1 + \sum_{\substack{u \in T(v)\\ u \neq v}}{\frac{g(u)}{d(u)}} = \frac{d(v)+|T(v)|}{d(v)+1} = g(v).$$

    To summarize, we just showed $$g(v) = 1+\sum_{\substack{u \in T(v)\\ u \neq v}}\frac{g(u)}{d(u)}.$$
    Multiplying both sides by $d!$ completes the proof.
\end{proof}

The proportional allocation algorithm is formally presented as below.

\def\z{
\medmuskip=1mu
\thinmuskip=1mu
\thickmuskip=1mu
}

\vspace{1em}
\begin{algorithm}
  \caption{\AlgTreeLike$(G)$}
  \label{alg:tree-like}
  \begin{algorithmic}[1]
    \REQUIRE  Descendant graph $G$ of tree $T$ and
    root vertex $r$, size of tree $n$, depth of tree $d$

     \FOR {$u \in T$ in increasing order of $d(u)$}
         \item[] // {\Large \strut} collect and cut phase:
         \STATE \label{step:collect}
         Let $X^u$ be the union of all cake pieces that vertex $u$ possesses
        %  $v$ collects her $f(v)$ pieces obtained from her ancestors, and slices into $f(v)$ equal pieces according to $v$'s evaluation function.
         \STATE \label{step:slice}
         Slice $X^u$ into $f(u)$ equal pieces $X^u_1, \ldots X^u_{f(u)}$ according to $u$'s evaluation function
         \item[] // {\Large \strut} distribute phase:
         \FOR {$v \in T(u)-\{u\}$ in increasing order of $d(v)$}
             \STATE Among all the remaining pieces of $X^u_1, \ldots, X^u_{f(u)}$, agent $v$ takes $f(v)/d(v)$ pieces that she values the highest.
         \ENDFOR \label{step:redistribute}
	\ENDFOR
%        \STATE Among all remaining pieces of $A_1, \ldots, A_n$, pick $|T(i)|$
%        pieces that agent $i$ values the highest.
%        \STATE Let $S_i$ denote the union of these $|T(i)|$ pieces.
%    \STATE \label{step: to root} Allocate the remaining one piece to $r$.
%    \FOR {each child $i$ of root $r$}
%        \STATE Apply \Cut$(i, r, |T(i)|, S_i)$ to cut $S_i$ into $|T(i)|$ equal
%        pieces $X_1, \ldots, X_{|T(i)|}$ for both $i$ and $r$.
%        \STATE Run \AlgTree$(T(i), i, (X_1, \ldots, X_{|T(i)|}))$.
%    \ENDFOR
  \end{algorithmic}
\end{algorithm}
\vspace{1em}

Figure~\ref{fig:fig2} illustrates the $f$ values and the execution of the
algorithm on
a simple 5-node tree.

\begin{figure*}[!ht]
	\centering
    \subfigure[$f$ values]{
        \label{fig:fvalues}
        \includegraphics[height=1.3in]{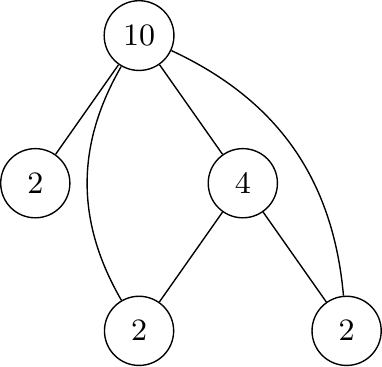}
    }
    % \quad
    \subfigure[execution]{
        \label{fig:exec}
        \includegraphics[height=1.3in]{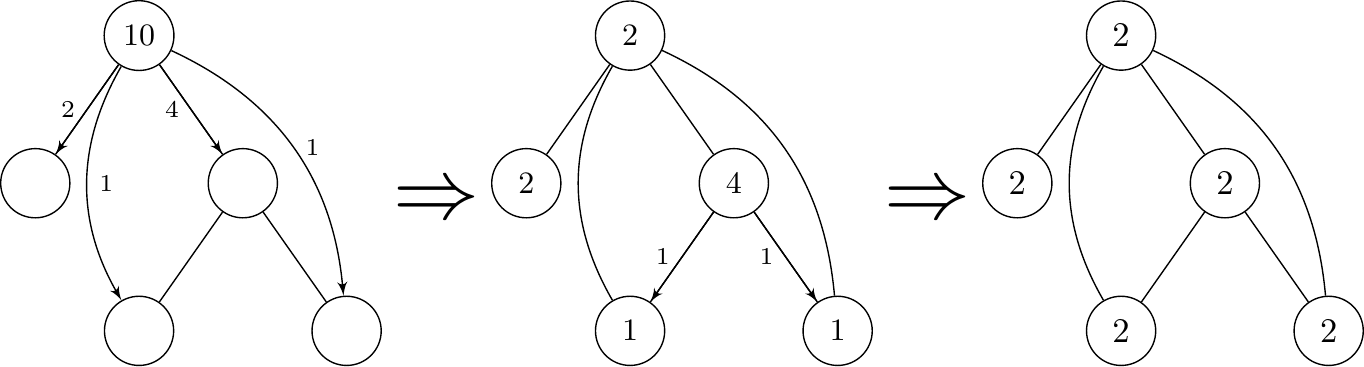}
    }
	\caption{$f$ values of and execution of Algorithm \ref{alg:tree-like} on a simple 5-node tree.}
  \label{fig:fig2}
\end{figure*}

From the algorithm description and equations~(\ref{eq:def_fv})
and~(\ref{eq:induct_fv}), it is easy to observe the following properties on the
number of cake slices during the algorithm.

\vspace{1em}

\noindent\hrulefill

\noindent{\bf Properties:}
  % \vspace{-0.15cm}
  % Algorithm \AlgTreeLike has the following properties:
  \begin{enumerate}
    \item[(1)] Every agent $v$, except the root of the tree, received $f(v)$ pieces of the cake in total from her ancestors.
    \item[(2)] Every agent will have exactly $d!$ slices of cake in her possession after the algorithm terminates.
  \end{enumerate}
\noindent\hrulefill
\vspace{0.5em}

%Property (1) tells us that at the collect and cut phases, the agents only re-balance the slices that they have received without changing the number of their pieces. Hence, after the algorithm terminates, the total number of slices of all agents remains to be $f(r)$, which has value $f(r) = |T|\cdot d!$.
%Property (2) means that every agent will have the same number of cake slices at vthe end of the algorithm. These two properties could help agents in comparing their cake share to that of their neighbors, as we will see in the proof of the theorem below.

\begin{Thm}
  For any $n$-node descendant graph $G$ of a tree, \AlgTreeLike$(G)$ outputs an
  allocation that is proportional on $G$ with at most $n^2\cdot d!$ cuts.
\end{Thm}
\begin{proof}

% Suppose we are about to execute line~\ref{step:reslice} for a vertex $v$ of depth
% $k$. Let the ancestors of $v$ be $u_0$, $u_1$, \dots, $u_{k-1}$ where $u_i$ is of
% depth $i$. Note that $u_0=r$. $v$ has obtained $f(v)/k$ slices from each of $u_i$,
% so $v$ has $f(v)$ slices in total.
%
% After executing line~\ref{step:redistribute}, by Proposition~\ref{prop:induct_fv},
% $v$
% has $d!$ slices left. In particular, this implies that every vertex will
% eventually have $d!$ slices in her possession after the algorithm terminates. Also
% note that the total number of slices remain invariant. This
% implies that just before line~\ref{step:collect} was executed, the number of
% slices distributed to $v$, her ancestors, and her descendants, is
% $N:=(k+|T(v)|)\cdot d!$.

% Consider the moment when line~\ref{step:collect} is about to be executed.
First note that the root $r$ cuts the cake into $|T|\cdot d!$ equal pieces and finally gets $d!$ of them, so its value is $1/|T|$ fraction of that of the whole cake. As the root connects to all nodes in the graph, it achieves exactly the average of its neighbors.

Now we consider an arbitrary node $v$ of depth at least 1. Let $N(v)$ be the set
of $v$'s neighbors. Furthermore, let $N_a(v)$ be the set of $v$'s ancestors in
$T$, and $N_d(v)$ the set of $v$'s descendants in $T$. It is clear that
$N(v)=N_a(v)\cup N_d(v)$. Also note that, as $G$ is a descendant graph, $N_d(v) =
T(v)
\setminus \{v\}$.
%Node $v$'s neighbors $N(v)$ consists of two parts:
%all $v$'s ancestors in tree $T$, which we denote as $N_a(v)$, and $v$'s
%descendants in $T$, which we denote as $N_d(v)$. Because $G$ is tree-like, we
%know $N_d(v) = T(v) / \{v\}$.

Let $A_v$ be the final allocation of agent $v$ at the end of the algorithm. We are
interested in $S = \bigcup_{i \in N(v)}A_i$, the union of all pieces of cake
belonging to agents in $N(v)$. Note that $S$ consists of the following two
components:
\begin{itemize}
  \item {\em Those held by $v$'s ancestors:} Each $A_u$ for $u \in N_a(v)$
  contains $d!$
  slices from $\{X^u_i\}_{i=1,\ldots,f(u)}$. These are the leftover slices after
  the distribute phase of agent $u$.
  \item {\em Those held by $v$'s descendants:} $\bigcup_{w \in N_d(v)}A_w$
  consists of two parts:
  \begin{enumerate}%[wide, labelwidth=!, labelindent=0pt]
    \item  The slices from $\{X^u_i\}$ distributed
    by agent $u \in N_a(v)$ in their distribute phase to agents in $N_d(v)$.
    For each $u \in N_a(v)$, each agent $w \in N_d(v)$ picks $\frac{f(w)}{d(w)}$ slices from $\{X^u_i\}$.
    In total, agent $u$ distributes $\sum_{w \in N_d(v)}\frac{f(w)}{d(w)}$ slices of cake to agents in $N_d(v)$.
    From Eq~\eqref{eq:induct_fv}, this number equals to $f(v) - d!$.
    \item The $f(v) - d!$ slices from $\{X^v_i\}$ distributed by agent $v$ in her distribute phase to her descendants. Note that $X^v$ comes from all ancestors $u$.
  \end{enumerate}
\end{itemize}

Note that each ancestor $u$ of $v$ has her collection $X^u$ distributed first to
$v$, then to $v$'s descendants, and finally to $u$ herself. Denote these parts by
$S_v^u$, $S_d^u$ and $S_u^u$, respectively. Agent $v$ later distributes exactly
$\frac{f(v)-d!}{f(v)}$ fraction of $S_v^u$ to its descents, and keeps
$\frac{d!}{f(v)}$ fraction of $S_v^u$ to herself. Let us denote these two parts by
$S_{vd}^u$ and $S_{vv}^u$, respectively. To avoid potential confusion of notation,
we use $\alpha_v$ to denote the valuation function of agent $v$. We will show the
following inequality.
\begin{align}\label{eq:u-proportional}
	\alpha_v(S_{vv}^u) \ge \frac{\alpha_v(S_u^u) + \alpha_v(S_d^u) + \alpha_v(S_{vd}^u)}{d(v)+T(v)-1}.
\end{align}
Once this is shown for all ancestors $u$ of $v$, we can sum these inequalities over all ancestors $u$ and obtain
\begin{align}\label{eq:proportional}
\alpha_v(A_v) \ge \frac{\alpha_v(S_a) + \alpha_v(S_d) + \alpha_v(S_{vd})}{d(v)+T(v)-1},
\end{align}
where we used the facts that
\begin{itemize}
	\item $A_v = \uplus_u S_{vv}^u$ is what $v$ finally has, where the notation $\uplus_u$ stands for the disjoint union over all ancestors $u$ of $v$.
	\item $S_a = \uplus_u S_u^u$ is the part of cake that $v$'s ancestors collectively have,
	\item $S_d = \uplus_u S_d^u$ is the set of pieces that $v$'s ancestors give to $v$'s descendants, and
	\item $S_{vd} = \uplus_u S_{vd}^u$ is the set of pieces that $v$ gives to its descendants.
\end{itemize}
Note that the numerator in the right hand side of Eq. \eqref{eq:proportional} is exactly the total value of $N(v)$, and the denominator is exactly the size of $N(v)$. Thus the inequality is actually \[\alpha_v(A_v) \ge \alpha_v(N(v))/|N(v)|,\] as the proportionality requires.

\medskip
So it remains to prove Eq. \eqref{eq:u-proportional}. We will examine the four sets involved in this inequality one by one, and represent or bound them all in terms of $\alpha_v(S_v^u)$.
\begin{itemize}\setlength\itemsep{1em}
	\item $\alpha_v(S_{vv}^u) = \alpha_v(S_v^u) \cdot d!/f(v)$, as $v$ divides $S_v^u$ into $f(v)$ equal pieces and takes $d!$ of them.
	\item $\frac{\alpha_v(S_u^u)}{d!} \le \frac{\alpha_v(S_v^u)}{f(v)/d(v)}$, as $v$ takes $f(v)/d(v)$ pieces of $\{X_i^u:i\in [f(u)]\}$ \emph{before} $u$ is left with $d!$ pieces.
	\item $\frac{\alpha_v(S_d^u)}{f(v)-d!} \le \frac{\alpha_v(S_v^u)}{f(v)/d(v)}$, as $v$ takes $f(v)/d(v)$ pieces of $\{X_i^u:i\in [f(u)]\}$ \emph{before} its descendants collectively take $(f(v)-d!)$ pieces.
	\item $\alpha_v(S_{vd}^u)= \alpha_v(S_v^u) \cdot (f(v)-d!)/f(v)$ as $v$ divides $S_v^u$ into $f(v)$ equal pieces and pass $(f(v)-d!)$ of them to descendants.
\end{itemize}
Putting these four (in)equalities and the definition of $f(v)$ in
Eq.~\eqref{eq:def_fv} together, one can easily verify
Eq.~\eqref{eq:u-proportional}. This completes the proof of the proportionality.

For the number of cuts, each agent $v$, when being processed, makes $f(v)$ cuts in the cut phase. In total the algorithm requires $\sum_{v}f(v) \leq n\cdot f(r) \leq n^2\cdot d!$ cuts.
\end{proof}

Note that though the number of cuts required here is exponential, this singly exponential bound is much better than the one for the general protocol in \cite{aziz2016discrete}.

\section{Conclusion}
This paper introduces a graphical framework for fair allocation of divisible good, defines envy-freeness and proportionality on a graph, and proposes an envy-free allocation algorithm on trees and a proportional allocation algorithm on descendant graphs. The framework opens new research directions in developing simple and efficient algorithms that produce fair allocations under important special graph structures.

\section*{Acknowledgments}
We thank Ariel Procaccia and Nisarg Shah for pointing out~\cite{chevaleyre2007allocating,todo2011generalizing}, and an anonymous reviewer for pointing out~\cite{abebe2017fair} to us.

This work was supported by
Australian Research Council DECRA DE150100720 and Research
Grants
Council of the Hong Kong S.A.R. (Project no. CUHK14239416).
\newpage

% \fontsize{9.5pt}{10.5pt}
% \selectfont
\bibliographystyle{plain}
\bibliography{refs}

\end{document}

%% file: intro.tex
%!TEX root = AllocationNetworks.tex
\newcommand{\shengyu}[1]{{\color{blue} #1}}
\section{Introduction}
In a nutshell, economics studies how resources are managed and allocated~\cite{mankiw2014principles}, and %resource allocation is also one of the most fundamental topics in computer science, e.g. in tiem slicing in multitask on single processor, bandwidth allocation among competing flows in networks, etc.
one of the most fundamental targets is to achieve certain fairness in the allocation
of resources. In a standard setting, different people have possibly different
preferences on parts of a common resource, and a fair allocation aims to
distribute the resource to the people so that everyone feels that she
is treated ``fairly''. When the resource is divisible, the problem is also known as ``cake cutting'',
and two of the most prominent fairness notions in this domain are \emph{envy-freenss} and \emph{proportionality}.
Here an allocation is envy-free if no agent $i$ envies
any other agent $j$; formally, this requires that for all $i,j\in [n]$, $v_i(A_i)
\ge v_i(A_j)$, where $v_k$ is the valuation function of agent $k$ and $A_k$ is
the part allocated to agent $k$.
It is known that an envy-free allocatio always exists~\cite{brams1996fair}.
However, it is until recently that a finite-step
procedure to find an envy-free allocation was proposed~\cite{aziz2016discrete}.
A weaker fairness solution concept is called
\textit{proportionality}, which requires that each agent gets at least the
average of the total utility (with respect to her valuation). Envy-freeness
implies proportionality, but not vice versa.

When studying these fairness conditions, almost all previous works consider all (ordered) pairs of relations among
agents. However, in many practical scenarios, the relations that need to be
considered are restricted.
Such restrictions are motivated from two perspectives:
(1) the system often involves a large number of people with an underlying social network structure, and most people are not aware of their non-neighbors' allocation or even their existence. It is thus inefficient and sometimes meaningless to consider any potential envy between pairs of people who do not know each other;
(2) institutional policies may introduce priorities among agents, and envies between agents of different priorities will not be considered toward unfairness. For example, when allocating network resources such as bandwidth to users, priorities will be given to users who paid a higher premium and they are expected to receive better shares even in a considered ``fair'' allocation.

% Most comparisons that cause our true envies are with people around us, such as neighbors, colleagues, or friends.
To capture the most salient aspect of the motivations above, in this paper, we initialize the study of fair cake cutting
 on graphs. Formally, we consider a graph $G$ with vertices being the
agents. An allocation is said to be \emph{envy-free on graph $G$} if no agent envies any of her neighbors in
the graph. An allocation is \emph{proportional on graph $G$} if each agent gets at least the average
of her neighbors' total allocation (with respect to her own valuation).
By only considering the fairness conditions between connected pairs, the
problem could potentially admit fair allocations with more desired properties such
as efficiency, and lead to simpler and more efficient algorithms to produce them.

With regard to envy-freeness, one can easily see that an envy-free allocation in
the traditional definition is also an envy-free allocation on any graph. However,
the algorithm to compute such a solution to the former~\cite{aziz2016discrete} has
extremely high complexity in terms of both the number of queries and the number of
cuts it requires. One goal of this line of studies is to design simple and
easy-to-implement algorithms for fair allocations on special graphs.
It is easy to observe that (cf. Section~\ref{subsec:motivation}), an envy-free
allocation protocol on a connected graph $G$ induces an envy-free protocol on any
of its spanning trees. Therefore, trees form a first bottleneck to the design of
envy-free algorithms for more sophisticated graph families.
%One important family of simple graphs is that of trees.
The first result we will
show is an efficient moving-knife algorithm to find an envy-free allocation on an
arbitrary tree, removing this bottleneck.
The procedure allocates the cake from the tree root in a top-down fashion and is
significantly simpler than the protocol proposed by Aziz et
al.~\cite{aziz2016discrete}.
% Note, however, that this is not a finite step protocol as it requires a continuous {\em moving knife} procedure.

When the graph is more than a tree, the added edges significantly increase the
difficulty of producing an envy-free allocation. However, if we lower the requirements and only aim at proportionality, then it is possible to go beyond trees.
Note that though a proportional allocation can be produced for complete graphs by several algorithms, such globally proportional solutions may be no longer proportional on an incomplete graph.
Our second result gives an discrete and efficient algorithm to find a proportional allocation
on descendant graphs, a family of graphs that are generated from rooted trees by
connecting all ancestor-descendant pairs in the tree.
%This family of graphs is interesting, because any graph can be embedded as a
%subgraph of a descendant graph
Descendant graphs are interesting from both practical and theoretical viewpoints
(cf. Section~\ref{subsec:motivation}).

\subsection{Related Work}
Cake cutting has been a central topic in recource allocation for decades; see, e.g.,~\cite{brams1996fair,RW98,Pro13}. As mentioned, envy-freeness and proportionality are two of the most prominent solution concepts for fairness consideration in this domain.
An envy-free allocation always exists, even if only $n-1$ cuts are
allowed~\cite{su1999}. The computation of envy-free and proportional allocations
has also received considerable attention, with a good number of discrete and
continuous protocols designed for different settings
~\cite{DS61,RW98,BT95,brams1997moving,barbanel2004cake,procaccia2009thou,procaccia2015cake,aziz2015discrete}.
It had been a long standing open question to design a discrete and
bounded envy-free protocol for $n$ agents, until settled very recently by Aziz and
Mackenzie~\cite{aziz2016discrete}.
For proportionality, much simpler algorithms are known that yield proportional
allocations with a reasonable number of cuts. Even and
Paz~\cite{even1984note}
gave a divide-and-conquer protocol that could produce a proportional allocation
with $O(n\log{n})$ cuts. Woeginger and Sgall~\cite{woeginger2007complexity}
and Edmonds and Pruhs~\cite{edmonds2006cake} later also showed a matching
$\Omega(n\log{n})$ lower bound.

% Besides algorithmic focus, researchers have also view cake cutting from game-theoretic perspectives and try to design truthful mechanisms that produce envy-free or proportional allocations. Both positive results~\cite{CLPP11,MN12} and negative results~\cite{schummer1996strategy,aziz2014cake} were shown in this regard. Another focus has been on the efficiency maximization in cake cutting with fairness guarantees.

The idea of restrictive relations, i.e., people comparing the allocation only with their peers, have also been considered in scenarios other than cake cutting. Chevaleyre et al.~\cite{chevaleyre2007allocating} proposed a negotiation framework with a topology structure, such that agents may only trade indivisible goods with their neighbors, and the envy also could only happen between connected pairs.
Todo et al.~\cite{todo2011generalizing} generalized the envy-freeness to allow
envies between groups of agents, and discussed the combinatorial auction
mechanisms that satisfy these concepts.

Shortly before our work, Abebe et al~.\cite{abebe2017fair} proposed the same notion of fairness on networks. Despite this, the paper and ours have completely different technical contents. Abebe et al~.\cite{abebe2017fair} shows that proportional allocations on networks do not satisfy any natural containment relations, characterizes the set of graphs for which a single-cutter protocol can produce an envy-free allocation, and analyzes the price of envy-freeness. In this paper, we focus on designing simple and efficient protocols for fair allocations on special classes of graphs. The graph families considered in these two papers are also different.